\documentclass{ecai}

\usepackage{latexsym}
\usepackage{amssymb}
\usepackage{amsmath}
\usepackage{amsthm}
\usepackage{booktabs}
\usepackage{enumitem}
\usepackage{graphicx}
\usepackage{color}

\usepackage[small]{caption}
\usepackage{url}
\usepackage[hidelinks]{hyperref}

\usepackage{multirow}
\usepackage{subcaption}
\usepackage{verbatim}
\usepackage{multicol}
\usepackage{placeins}
\usepackage{makecell}
\usepackage{thm-restate}

\usepackage[ruled,vlined]{algorithm2e}
\usepackage {algpseudocode}
\usepackage{algorithmicx}
\usepackage{algcompatible}
\usepackage{cleveref}
\SetKwInput{KwInput}{Input}  
\SetKwInput{KwOutput}{Output}

\newtheorem{theorem}{Theorem}

\newtheorem{definition}{Definition}
\newtheorem{observation}{Observation}

\newcommand{\genmincostalgo}{\emph{Min-Cost}}
\newcommand{\mincostalgo}{\emph{Seq-Hungarian}}
\newcommand{\firstfairalgotwo}{\emph{C-Balance}}
\newcommand{\firstfairalgo}{\emph{DC-Balance}}
\newcommand{\extfairalgo}{\emph{EDC-Balance}}

\newcommand{\BibTeX}{B\kern-.05em{\sc i\kern-.025em b}\kern-.08em\TeX}

\usepackage{etoolbox} 
\newcommand{\proceedingsnote}{
\vspace{-1.5cm}
  Please cite this paper as: Vibulan J., Swapnil Dhamal, and Shweta Jain. ``The Multi-Stage Assignment Problem:
A Fairness Perspective.'' Proceedings of the 28th European Conference on Artificial Intelligence, 2025.
  \vspace{-0.5cm}
}

\begin{document}

\begin{frontmatter}

\paperid{0859}

\preto\maketitle{\begin{center}\proceedingsnote\end{center}\vspace{1em}}

\title{The Multi-Stage Assignment Problem:\\A Fairness Perspective}

\author[A]{\fnms{Vibulan}~\snm{J}\orcid{0009-0002-9445-3694}}
\author[B]{\fnms{Swapnil}~\snm{Dhamal}\orcid{0000-0001-7434-0778}}
\author[B]{\fnms{Shweta}~\snm{Jain}\orcid{0000-0002-2666-9058}} 

\address[A]{Indian Institute of Information Technology, Design and Manufacturing, Kancheepuram, Chennai, India}
\address[B]{Indian Institute of Technology Ropar, Rupnagar, India}

\begin{abstract}
This paper explores the problem of fair assignment on Multi-Stage graphs. A \emph{multi-stage graph} consists of nodes partitioned into $K$ disjoint sets (stages) structured as a sequence of weighted bipartite graphs formed across adjacent stages. The goal is to assign node-disjoint paths to $n$ agents starting from the first stage and ending in the last stage. We show that an efficient assignment that minimizes the overall sum of costs of all the agents' paths may be highly unfair and lead to significant cost disparities (envy) among the agents. We further show that finding an envy-minimizing assignment on a multi-stage graph is NP-hard. We propose the \emph{C-Balance} algorithm, which guarantees envy that is bounded by $2M$ in the case of two agents, where $M$ is the maximum edge weight. We demonstrate the algorithm's tightness by presenting an instance where the envy is $2M$. We further show that the cost of fairness ($CoF$), defined as the ratio of the cost of the assignment given by the fair algorithm to that of the minimum cost assignment, is bounded by $2$ for \emph{C-Balance}. We then extend this approach to $n$ agents by proposing the \emph{DC-Balance} algorithm that makes iterative calls to \emph{C-Balance}. We show the convergence of \emph{DC-Balance}, resulting in envy that is arbitrarily close to $2M$. We derive $CoF$ bounds for \emph{DC-Balance} and provide insights about its dependency on the instance-specific parameters and the desired degree of envy. We experimentally show that our algorithm runs several orders of magnitude faster than a suitably formulated ILP.
\end{abstract}

\end{frontmatter}

\label{page:start}

\section{Introduction}
\label{sec:intro}
Assignment of resources over graphs is crucial for optimizing the use of resources in various domains such as manufacturing, project management, parallel computing, and routing. In graph representation of the underlying problem, tasks, resources, or processes are often represented as nodes, and the dependencies between them as edges. Graphs systematically represent a problem's key features like dependencies, precedence constraints, and other relationships between tasks or processes. One such type of graph is a fully connected multi-stage (FCMS) graph, where the set of nodes is partitioned into stages, and every node in a stage is connected to every node in the next stage. FCMS graphs have been extensively deployed in supply chain \cite{bahrampour2016modeling}, vehicle routing problem \cite{KIM20108424},  task assignment \cite{VALADARESTAVARES199092}, 
routing optimization, portfolio optimization and financial strategies \cite{MULVEY20041}, etc. 
Note that the full-connectedness across adjacent stages is justified in most applications, since in any connected graph, even if there is no direct edge between two nodes, one can find the shortest path cost between these two nodes using a shortest path algorithm and connect these nodes with a direct edge having the computed cost. 

In FCMS graph, assigning resources/tasks to optimize a utility objective, such as a function of cost and time, has been a problem of immense interest, and there exist many exciting works in the literature. However, apart from maximizing utility in real-world applications, the fairness of resources among agents is critical. Fair assignment refers to assigning resources or workloads equitably  
while ensuring that no participant is unfairly burdened or underutilized.

To understand why fair assignment on FCMS graph is important, consider the example of multi-product multi-stage supply chain network design \cite{bahrampour2016modeling}. A company like Bosch manages multiple clients (agents) like Toyota, Hyundai, Honda, Ford, BMW, etc. Consider a case where each client requires Bosch to produce an automobile part that follows a given manufacturing process using assembly lines available. A stage represents a manufacturing stage of an automobile part, and nodes represent multiple workstations at several locations. Since each part must pass through each manufacturing stage, it must be allocated exactly one workstation in each stage. Traveling from one workstation to another is costly and time-consuming, and traditionally, this problem has been looked at from the lens of minimizing the total cost/time.
While it is essential to minimize the overall cost/time that Bosch takes to deliver the automobile parts to the different clients, it is also essential to minimize the difference in product delivery costs/times across the clients, for instance, to ensure fair market competition among them. This corresponds to minimizing envy in the multi-stage graph assignment problem. 

It is also important that in each assembly line, production capacity is optimally utilized, and no workstation is overloaded. Therefore, we look for a solution where, at each stage, each node can be assigned to 
not more than
one agent. 
Further, as our primary motivation comes from the application of multi-product multi-stage supply chain network design, 
we assume that the number of nodes at each stage is greater than or equal to the number of agents.
This is because the products are similar (therefore having to go through all the stages), and if there are fewer nodes at any stage, it implies that there exists a node representing a large facility that can accommodate multiple products (let’s say $b$) at a given time. In that case, one can create $b$ replicas of this node, thus effectively resulting in the number of nodes at each stage being greater than or equal to the number of agents.

To summarize, we consider fair assignment on a fully connected multi-stage (FCMS) graph.
For ease of exposition, we address the problem on a class of FCMS graphs wherein the number of nodes in each stage is equal to the number of agents; 
we term this problem as fair assignment on balanced fully connected multi-stage (BFCMS) graph.
The problem involves assigning each agent, such as a client in the supply chain example, 
a non-overlapping path on the BFCMS graph. For each assigned path, the agent incurs a cost, and the goal is to mitigate the difference in total cost incurred by any two agents. \Cref{fig:MLGR} shows an example of BFCMS and a valid assignment on the same.
We later show that our approach and results are directly applicable to general FCMS graphs. 
Next, to understand the effects of reducing the cost disparities (envy) among the agents, we study the ``cost of fairness ($CoF$)" of the proposed fair algorithms, which is defined as the ratio of the cost of the assignment provided by the fair algorithm to that of the minimum cost one.

\begin{figure}[t]

\begin{subfigure}[b]{0.24\textwidth}
    \centering\includegraphics[width=0.85\linewidth]{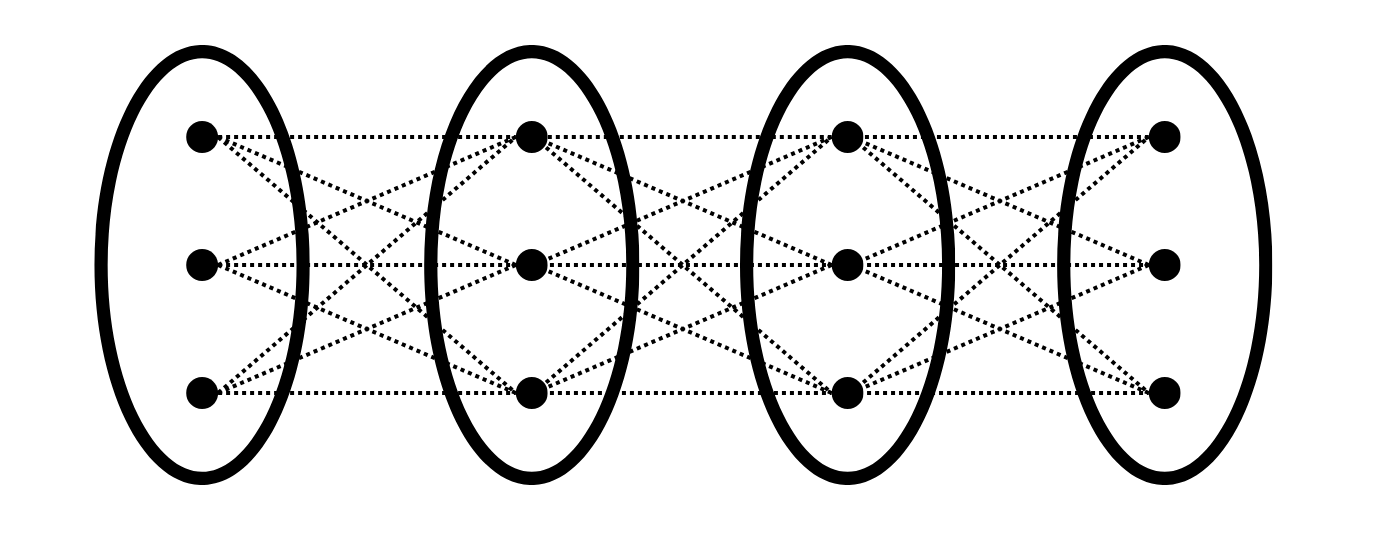}
    \caption{}
    \label{fig:multilevelgraph}
  \end{subfigure}
  \begin{subfigure}[b]{0.23\textwidth}
    \centering \includegraphics[width=0.85\linewidth]{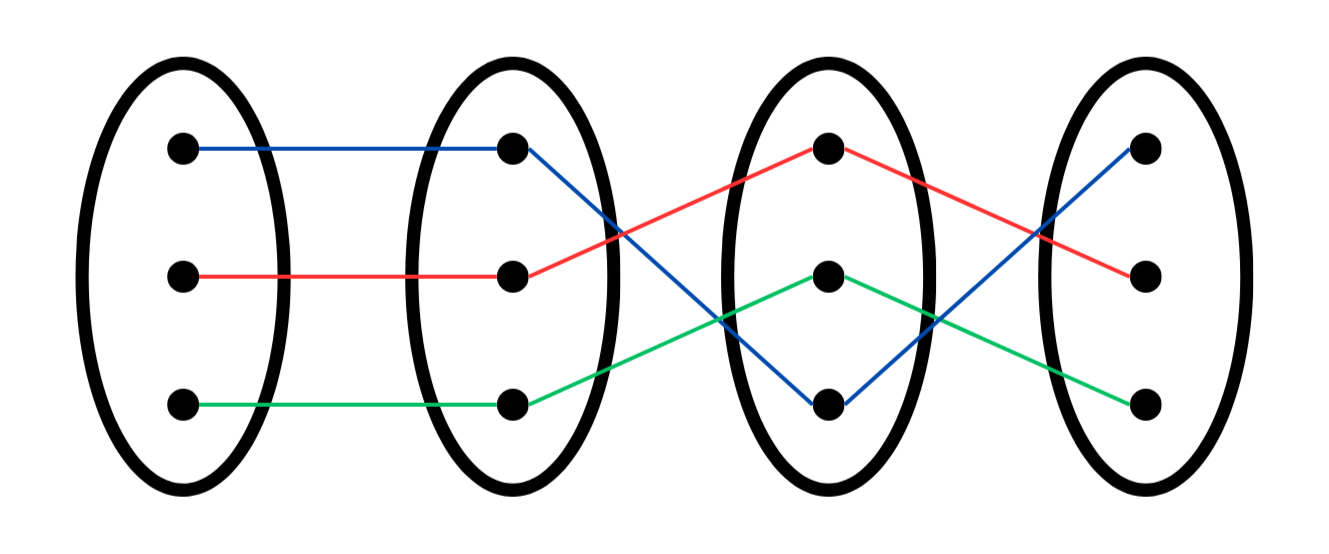 }
    \caption{}
    \label{fig:validsol}
  \end{subfigure}
  \captionsetup{skip=-12pt}
  \caption{(a) A BFCMS graph (edge weights not shown for better readability). (b) A valid assignment on the BFCMS graph.}
\label{fig:MLGR}
\end{figure}

\subsection*{Our Contributions and Results}

This paper aims to solve the fair assignment problem for $n$ agents on a given FCMS graph with $K$ stages, with $M$ denoting the maximum edge weight in the graph. Our primary contributions are as follows:

 \begin{itemize}[leftmargin=*]
 \setlength\itemsep{0.5em}

\item 
We start by showing the following:
\begin{itemize}[topsep=0.5em, leftmargin=*]
\item A minimum cost solution may result in an assignment with significantly high envy [\Cref{obs:unfairsol}].

\item Finding an envy-minimizing
assignment on an arbitrary BFCMS graph with a given $K \geq 4$ is NP-hard [\Cref{thm:nphard_numagents}],
and so is that with a given $n \geq 2$
[\Cref{thm:nphard_numlevels}]. 
\end{itemize}

    \item For the case involving two agents, we propose the Cost-Balance (\firstfairalgotwo) algorithm, which swaps the assignments to the agents in the minimum cost solution from a strategically chosen stage onwards.
    We show that:
\begin{itemize}[topsep=0.5em, leftmargin=*]
    \item \firstfairalgotwo\ ensures that the envy among agents is bounded by $2M$ [\Cref{thm:envycbalance}],

    \item The achieved bound of $2M$ is tight [\Cref{thm:cbalancetightbound}],

    \item The $CoF$ for \firstfairalgotwo\ is bounded by $2$ [\Cref{thm:cbalancecof}].
\end{itemize}
    
    \item Considering the general case with $n$ agents, we propose the Dynamic Cost-Balance (\firstfairalgo) algorithm that makes iterative calls to \firstfairalgotwo , and bounds the envy arbitrarily close to $2M$.
Given $\alpha>0$ quantifying the desired degree of envy, we show:
\begin{itemize}[topsep=0.5em, leftmargin=*]
    \item The number of iterative calls required in order to achieve envy of $(2+\alpha)M$ is $O\left(n\log\left(\frac{K}{\alpha}\right)\right)$ [\Cref{thm:num_iters}],
    
\item The $CoF$ is 
$O\left(n\log\left(\frac{K}{\alpha}\right)\right)$
[\Cref{thm:CoF}],

\item For minimizing envy,
it is infeasible to provide a constant instance-independent bound on $CoF$
[\Cref{obs:noinstind}].
\end{itemize}

\item 
We experimentally show that our algorithm runs several orders of magnitude faster than the 
ILP formulated
to minimize the total cost with the constraint that envy is bounded by $2M$.

\end{itemize}

\section{Related Work}
Addressing 
the problem of minimizing envy
has been a topic explored in prior research, such as fair division of divisible goods/tasks \cite{aziz2016,dubins1961,steinhaus1949}, assignment of indivisible goods/tasks \cite{amanatidis2022fair,bei2021,bhaskar2021approximate,lipton2004}, along with 
applications including group trip planning \cite{solanki2023fairness}, vehicle routing \cite{aleksandrov2024driver,matl2017workload}, fair bandwidth allocation \cite{coniglio2020}, fair routing for underwater networks \cite{diamant2018}, fair cost allocation in ride-sharing services \cite{lu2019}, fairness-aware load balancing \cite{kishor2020}, fair delivery \cite{hosseini2023fair}, and more.

Several works in fair delivery settings \cite{lopezsanchez2022,esmaeili2023rawlsian,nair2022gigs,gupta2022} focused on empirical studies without any theoretical guarantees. Recently, \cite{hosseini2023fair} provided some theoretical guarantees on a general unweighted tree graph with fairness computed on the number of hops each agent travels. While this is valuable, fairness considerations must go beyond mere order counts 
and account for the cost incurred (e.g., distance traveled) for completing the orders,
when extending the scenario to a city-scale delivery network
because the costs incurred for completing a given number of orders could vary significantly.
With respect to fairness in vehicular routing problems, a few works provide a fair solution by formulating the problem as an Integer Linear Program (ILP) \cite{lopezsanchez2022}. We primarily focus on delivering efficient algorithmic solutions to the problem with provable bounds by connecting the problem to that of fair routing planning \cite{hashem2015,solanki2023fairness}.

 Efforts have been made to introduce fairness in group trip planning  \cite{singhal2022,solanki2023fairness}, which aims at selecting a common path for all the agents as opposed to ours; in routing, on achieving fairness in terms of traffic on each edge \cite{kleinberg1999fairness,pioro2007fair}, and achieving collective and individual fairness among agents \cite{lopezsanchez2022,esmaeili2023rawlsian}. Our work distinguishes itself by focusing on individual fairness among agents, and since each agent follows a unique route (sequence of nodes assigned in the different stages), there is no need to address fairness in terms of traffic on each edge. Several studies have also explored supply chain optimization, focusing on efficiency improvements across different segments \cite{geunes2006supply,narahari2007bayesian,garg2005groves,biswas2004object,narahari2007performance}. Notably, the focus on fairness within this domain primarily revolves around fair profit distribution among participants \cite{chen2003multiobjective,yue2014fair,sawik2015fair,qiu2023pipeline}, which is different from our work.

The problem of fairness in FCMS graphs resembles that of fair allocation of items/tasks on graphs under `a connected and equal number of goods' constraints by representing each node as an item/task \cite{bouveret2017fair,misra2022a}. In FCMS, however, the costs of the tasks in one stage depends on the allotted tasks in the previous stage, as opposed to 
being fixed.
This constraint also makes our problem different from the problem of repeated allocation over time \cite{igarashi2023repeated,caragiannis2023repeatedly,aleksandrov2024limiting,endriss2013reduction} as these works do not consider 
the setting where the cost of a task at a time depends on the earlier allocated task. Similarly, in a min-max tree cover problem \cite{even2004min}, the goal is to find $n$ distinct trees 
with the objective of minimizing the maximum weight of any tree.
Note that an FCMS assignment demands that each tree must have exactly one node from each stage and each node must belong to at most one tree.
This constraint is not present in a general tree cover problem.

\section{Preliminaries}
Consider a weighted graph $\mathcal{G} = (V, E)$, where $V$ denotes the set of nodes and $E$ is the set of edges. The set $V$ is partitioned into $K$ disjoint subsets (or stages) $\{V_1, V_2, \ldots, V_K\}$, each containing $n$ nodes. A balanced fully connected multi-stage (BFCMS) graph (as depicted in \Cref{fig:multilevelgraph}) is defined as a sequence of $K$ such stages where a node from the $j$-th stage is linked to every node in the $j+1$-th stage through a weighted edge $e$ with weight $w_e$. 
We further denote $M = \max_{e \in E} w_e$ as the maximum edge weight in this graph. 

A valid solution $S$ to the assignment in 
BFCMS problem involves finding $n$ node-disjoint paths, each originating from a node in stage $V_1$ and terminating at a node in stage $V_K$. The solution is expressed as $S = \{P_1, P_2, \ldots, P_n\}$, where each $P_i = \{p_{i}^1, p_{i}^2, \ldots, p_{i}^{K-1}\}$ signifies an individual path with edges $p_i^j$ from stage $j$ to $j+1$ for an agent $i$. Let us denote the set of all valid solutions by $\mathcal{F}$. It should be noted that each agent $i$ is assigned one path $P_i$, meaning that the cost incurred on path $P_i$ also represents the cost incurred by agent $i$.  
\Cref{fig:validsol} depicts a valid solution to the assignment in BFCMS problem with three agents and four stages.

Note that while we consider node-disjoint paths, the physical routes themselves might not be entirely separated. 
It is important to make a distinction of the considered multi-stage assignment problem from the routing problem.
For instance, in the supply chain example discussed earlier, a path is defined by the set of workstation nodes that are to be assigned to an agent across the different stages, and not by the set of endpoints of road segments in the underlying physical road network.
If an agent is assigned a node, no other agent can be assigned that node, thus resulting in the paths being node-disjoint.

\begin{definition}[Cost of a Solution]
    The cost of a solution $S \in \mathcal{F}$ with $S=\{P_1, P_2, \ldots, P_n\}$ is defined as the sum of the weights of all the edges in the solution, expressed as $C(S) = \sum_{i=1}^n\sum_{e \in P_i} w_e$.

\end{definition}
    We also use the notation $C(P_i) = \sum_{e\in P_i} w_e$ to denote the cost of a particular path $P_i$ (or cost incurred by agent $i$).

\begin{definition}[Envy of a Solution]
    The envy of a solution $S\in \mathcal{F}$ is defined as the maximum difference between the cost of any two agents (paths) in the solution, expressed as $\mathcal{E}(S) = \max_{P_i,P_j \in S} \left( C(P_i) - C(P_j) \right)$.
\end{definition}

\begin{definition}[Cost of Fairness]
    The cost of fairness ($CoF$) for a fair algorithm $\mathcal{A}$ is defined as the ratio of the cost of the solution produced by $\mathcal{A}$ ($S_\mathcal{A}$) to the cost of the optimal solution ($S^* = \arg\min_{S\in \mathcal{F}}C(S)$), i.e.  
    $CoF(\mathcal{A}) = \frac{C(S_\mathcal{A})}{C(S^*)}$. 
\end{definition}

The notion of CoF is in line with the literature \cite{das2022individual}.
Note that it differs from Price of Fairness (PoF), which quantifies the worst-case ratio of the optimal `fair' solution to that of the optimal solution.

\subsubsection*{Minimum Cost Assignment}

\label{sec:mincostsol}

A minimum cost assignment can be obtained using 
an adaptation of Suurballe's algorithm \cite{suurballe1974disjoint} which finds $n$ node-disjoint paths of minimum total length from a source node to a terminal node. 
The adaptation involves making all the edges directed such that the direction is from a node in a stage to a node in the next stage, and
adding dummy source and terminal nodes with zero-weight edge from the source node to every node in the first stage as well as from every node in the last stage to the terminal node.
A minimum cost assignment is nothing but the set of $n$ node-disjoint paths thus obtained, post deletion of the dummy nodes and edges. 
If different agents have different pair of source and terminal nodes, then one can also use $O(\nu^3)$ algorithm provided in \cite{ROBERTSON199565} with $\nu$ being the total number of nodes in the graph.

 Note that a minimum cost solution
 can be highly unfair. 
 \Cref{fig:unfairsol_a}
 depicts such a scenario with $K$ stages and two agents; a unique optimal solution is marked with blue and red paths. In this case, the total cost as well as envy of the solution is $(K-1)(M-\delta)$. 
 With an arbitrarily high value of $K$ and especially if $\delta << M$, the envy will become arbitrarily high. Hence, we have the following observation.

 \begin{observation}
A minimum cost solution may result in an assignment with significantly high envy.
\label{obs:unfairsol}
 \end{observation}

 Note in the above example that if we swap the paths allocated to agents from stage $\lfloor K/2 \rfloor$ onwards, we obtain a new solution as depicted in 
 \Cref{fig:unfairsol_b},
 which has an envy less than $M$.
However, an additional cost is incurred while performing this swap in the form of considering suboptimal edges in place of the previously considered optimal ones.
In this example, the total cost rises to $(K-2)(M-\delta)+2M$. We next discuss the algorithms towards fair assignment on BFCMS while maintaining a bound on CoF.

\begin{figure}[t]
\hspace{-0.35cm}
\begin{subfigure}[b]{0.24\textwidth}
    \centering\includegraphics[width=\linewidth]{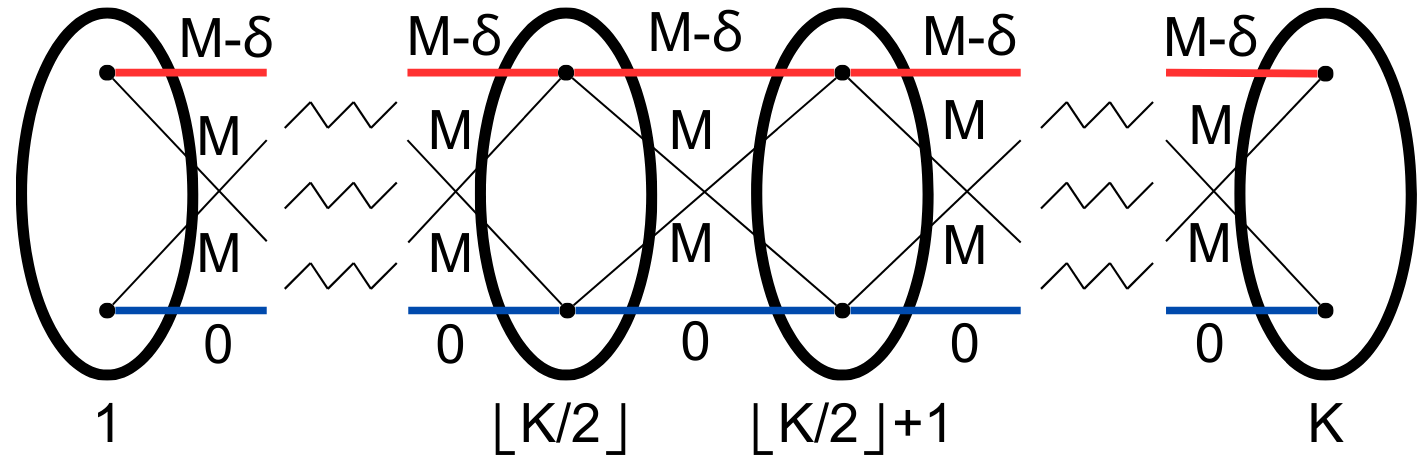}
    \caption{}
    \label{fig:unfairsol_a}
  \end{subfigure}
  \hspace{0.1cm}
  \begin{subfigure}[b]{0.24\textwidth}
    \centering \includegraphics[width=1\linewidth]{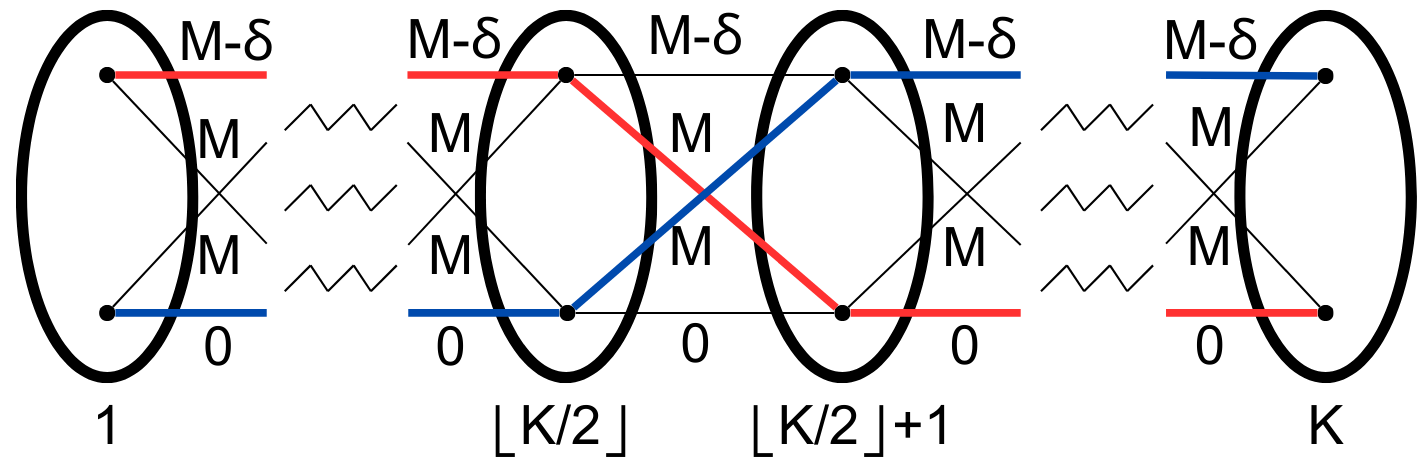 }
    \caption{}
    \label{fig:unfairsol_b}
  \end{subfigure}
    \captionsetup{skip=-12pt}
\caption{Illustration of (a) a minimum cost solution being highly unfair and (b) a possible workaround.}
\end{figure}

\section{Fair Assignment}

For ease of exposition, we present our approach and analysis while considering a BFCMS graph,
wherein the number of nodes in each stage equals the number of agents.
However, our approach and analysis are directly applicable to a general FCMS graph, wherein the number of nodes at some or all of the stages could be different from the number of agents; a note on this is provided in \Cref{sec:unequal_stages}.
We begin with showing hardness results for minimizing envy on BFCMS.

\subsection{NP-Hardness of Minimizing Envy}

\begin{restatable}{theorem}{thmnphardnumagents}
For an arbitrary BFCMS graph $\mathcal{G}$ with $K \geq 4$ stages, finding an envy minimizing assignment is NP-hard.
\label{thm:nphard_numagents}
\end{restatable}

\begin{restatable}{theorem}{thmnphardnumlevels}
For an arbitrary BFCMS graph $\mathcal{G}$ with $n \geq 2$ agents, finding an envy minimizing assignment is NP-hard.
\label{thm:nphard_numlevels}
\end{restatable}
The proofs provided in 
\Cref{app:nphard} 
involve reduction from Numerical 3-Dimensional Matching for \Cref{thm:nphard_numagents} and Multiway Number Partitioning for \Cref{thm:nphard_numlevels}. For both reductions, 
while constructing BFCMS instances, nodes are labeled with integers from the original problem. For all incoming edges to a node labeled $\ell$, the edge weight is set to $\ell$. So, a path's cost equals the sum of labels it traverses. Hence, all valid assignments result in a CoF of 1 
as they have the same total cost. 
So, the NP-hardness persists even when the search space is restricted to solutions satisfying a constant CoF bound (e.g., CoF $\leq$ 2).
We now present our fair algorithm for two agents, and subsequently, for $n>2$.

\begin{algorithm}[t]
\DontPrintSemicolon
  
  \KwInput{A multi-stage graph $\mathcal{G} (\{V_i\}_{i=1}^K, E)$, Solution $S^*=\{P_1,P_2\}$ with $C(P_1) > C(P_2)$ } 
  \KwOutput{Fair solution $S$}

 $C_1 \gets C(P_1)$, $C_2 \gets C(P_2)$,  $\epsilon \gets C_1-C_2$
     
\If{$\epsilon > 2M$}
{
    \For{each $i$ in $[K-1]$}
    {
     $C_1'' \gets \text{Cost of agent $1$ till stage $i-1$}$
     
     $C_2'' \gets \text{Cost of agent $2$ till stage $i-1$}$
    
     $C_1' \gets \text{Cost of agent $1$ till stage $i$}$ 
    
     $C_1' \gets \text{Cost of agent $2$ till stage $i$}$ 
    
    \If{$C_1'' - C_2'' \le \epsilon/2$ and $C_1'  - C_2' > \epsilon/2$}
    {
      $\text{Swap the node assignment in $S$ at this stage}$
      
      \textbf{break}
    }
    }
}
\caption{\firstfairalgotwo}
\label{alg:cbalance}
\end{algorithm}

\subsection{Cost-Balance Algorithm for Two Agents}
\label{subsec:cbalance}

In the simplified scenario with two agents, we identify a stage $i$ in $S^*$ obtained from the minimum cost assignment and perform a swap between the nodes allocated from this stage onwards such that the resulting envy is bounded by $2M$. An example of one such swap is provided in 
\Cref{fig:unfairsol_b}.
We term this algorithm as the Cost-Balance (\firstfairalgotwo) algorithm 
[\Cref{alg:cbalance}].

Without loss of generality, assume for $P_1, P_2\in S$: $C(P_1) - C(P_2) > 2M$. If this does not hold, \firstfairalgotwo\ simply returns $S^*$. Let $\epsilon = \mathcal{E}(S^*)$ and define $i$ to be the minimum stage such that the cost difference between $C(P_1)$ and $C(P_2)$ 
becomes greater than 
$\frac{\epsilon}{2}$. \firstfairalgotwo\ swaps the paths of agents $1$ and $2$ at stage $i-1$. 
Denote the costs of paths $P_1$ and $P_2$ up to stage $i-1$ by $C_1''$ and $C_2''$, and their costs up to stage $i$ by $C_1' = C_1''+ x$ and $C_2' = C_2'' + y$, respectively. 
Let
the swapping cost for path $P_1$ be $x'$ and for path $P_2$ be $y'$.  
\Cref{fig:cbalanceproof}
illustrates this entire scenario.
Since the jump in cost difference cannot exceed $M$ at any consecutive stages, we have:
\begin{align*}
& \frac{\epsilon}{2} - M \le C_1'' - C_2'' \le \frac{\epsilon}{2} ,\;\;
\text{and} \; 
& \frac{\epsilon}{2} < C_1' - C_2' \le \frac{\epsilon}{2}+M .
\end{align*}

\begin{figure}[t]
    \centering
    \includegraphics[width=0.8\linewidth]{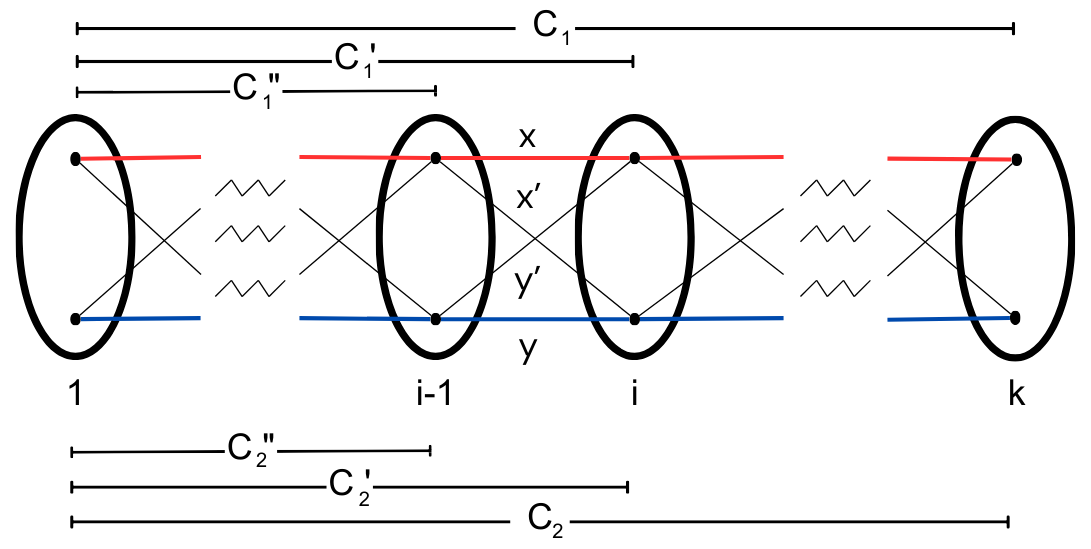}
    \caption{Illustration for the proof of \Cref{thm:envycbalance} showing that envy of the solution obtained by \firstfairalgotwo\ is bounded.}
    \label{fig:cbalanceproof}
\end{figure}

\begin{theorem}
The solution $S$ obtained by the \firstfairalgotwo\ algorithm has an envy of at most $2M$, i.e., $\mathcal{E}(S)\le 2M$.
\label{thm:envycbalance}
\end{theorem}
\begin{proof}
From 
\Cref{fig:cbalanceproof}, the envy between paths $P_1$ and $P_2$ after swapping the allocated nodes at stage $i$ is given by 
\begin{align*}
&\lvert C_1'' + x' + C_2 - C_2' - (C_2'' + y' + C_1 - C_1') \rvert
\\
&\le \lvert C_1'' - C_2'' + C_1' - C_2' - \epsilon \rvert + \lvert x' - y' \rvert
\le M + M = 2M .
\qedhere
\end{align*}
\end{proof}

\begin{theorem}
The bound given in 
\Cref{thm:envycbalance}
is tight, i.e., there exists an instance where envy is $2M$.
\label{thm:cbalancetightbound}
\end{theorem}
\begin{proof}
From 
\Cref{fig:extremecase},
if one of the agents is assigned the upper nodes in all the stages and the other agent is assigned the lower nodes, the total cost of the former's path is $M+M=2M$ and that of the latter's path is $0+0=0$, thus resulting in an envy of $2M$.
It can be similarly seen that in all potential path assignments, one of the agents incurs a cost of $2M$ and the other agent incurs a cost of $0$, resulting in an envy of $2M$.
So, the minimum envy achievable for this instance is $2M$, which is achieved by \firstfairalgotwo .
\end{proof}

\begin{theorem}
The cost of fairness for \firstfairalgotwo\ is bounded by $2$, i.e., $CoF(\firstfairalgotwo) < 2$.
\label{thm:cbalancecof}
\end{theorem}
\begin{proof}
When $C(S^*) > 2M$, a single swap at any stage can increase the cost of optimal solution $S^*$ by at most $2M$. So, 
\begin{align*}
CoF(\firstfairalgotwo) &= \frac{C(S^*) + 2M}{C(S^*)} = 1 + \frac{2M}{C(S^*)} < 2.
\end{align*}
When $C(S^*) \le 2M$, we have $\mathcal{E}(S^*) \le 2M$ and hence,  \firstfairalgotwo\ will not perform any swaps. In this case, $CoF(\firstfairalgotwo) = 1$.
\end{proof}

\begin{algorithm}[t]
\DontPrintSemicolon
  
  \KwInput{A multi-stage graph $\mathcal{G} (\{V_i\}_{i=1}^K, E)$, Solution $S^*=\{P_1,P_2,\ldots,P_n\}$ with $C(P_1) \ge C(P_2) \ge \ldots \ge C(P_n)$, Parameter $\alpha$} 
  \KwOutput{Fair solution $S$}

     $\epsilon \gets C(P_1) - C(P_n)$, $S' \gets \{P_1, P_n\}$
     
    \If{$\epsilon > M(2+\alpha)$}
    {
         $P_1, P_n \gets \text{C-Balance}(\mathcal{G}(\{V_i\}_{i=1}^K,E), S')$

         Rearrange agents' numbering such that $C(P_1) \ge C(P_2) \ge \ldots \ge C(P_n)$

         $S \gets \{P_1, P_2, \ldots, P_n\}$

         \firstfairalgo$(\mathcal{G}(\{V_i\}_{i=1}^K,E), S)$
    }

\caption{\firstfairalgo}
\label{alg:dcbalance}
\end{algorithm}

\begin{figure}[t]
    \centering
    \includegraphics[width=0.45\linewidth]{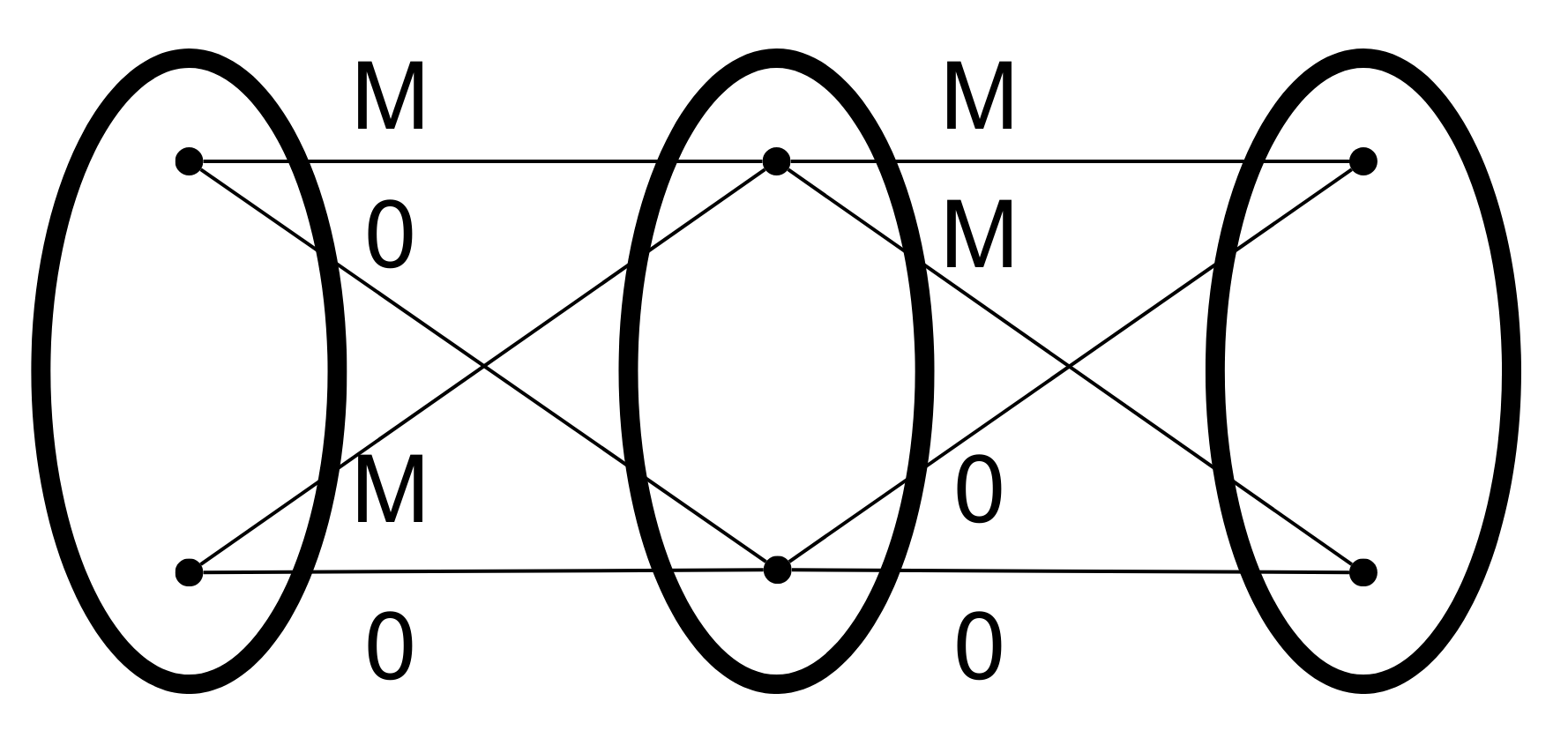}
    \caption{An instance with every possible assignment of paths resulting in an envy of $2M$.}
    \label{fig:extremecase}
\end{figure}

\subsection{Dynamic Cost-Balance Algorithm for \texorpdfstring{$n>2$}{n > 2}}
\label{subsec:DCalgo}

The ideas used in \firstfairalgotwo\ can be extended to the general case with $n$ agents. Here, in each iteration, we choose the agent with the minimum total cost and the agent with the maximum total cost, then perform a swap of assignment between them using the \firstfairalgotwo\ algorithm. This process is repeated until the overall envy
is bounded by $(2 + \alpha)M$, for any given $\alpha > 0$. This algorithm is termed the Dynamic Cost-Balance (\firstfairalgo) algorithm 
[\Cref{alg:dcbalance}].

Similar to the case of $n=2$, let the costs of paths $P_1$ and $P_n$ for the two agents with maximum and minimum costs in this assignment be denoted by $C_1$ and $C_n$, respectively. Let us further denote the costs of paths $P_1$ and $P_n$ up to stage $i-1$ by $C_1''$ and $C_n''$, respectively. The costs of paths $P_1$ and $P_n$ up to stage $i$ are denoted by $C_1' = C_1''+ x$ and $C_n' = C_n'' + y$, respectively. Additionally, we assume that the swapping cost for path $P_1$ is $x'$ and for path $P_n$ is $y'$. 
Similar to the two agents' case, if $\epsilon = C_1-C_n$, the choice of $i$ ensures that
\begin{align}
\frac{\epsilon}{2} - M \le C_1'' - C_n'' \le \frac{\epsilon}{2},  \; \text{and} \;\;
\frac{\epsilon}{2} < C_1' - C_n' \le \frac{\epsilon}{2}+M .
\label{eq:ineq_n}
\end{align}
We then have the following lemma,
which can be proved using the inequalities in (\ref{eq:ineq_n}),
alongwith 
the fact that
$-M \leq w_1-w_2 \leq M,$ $  \forall w_1, w_2 \in \{x, x', y, y'\} \; (\because x, x', y, y' \in [0,M])$.
A step-by-step proof is provided in 
\Cref{app:relaxation}.

\begin{restatable}{lemma}{lemmarelaxation}
If initial envy $C_1-C_n = \epsilon = 2M + 2\delta$, then after a swap, the new costs $C_1^{new}$ and $C_n^{new}$ will follow: $C_n + \delta \le C_1^{new} \le C_1 - \delta$ and $C_n + \delta \le C_n^{new} \le C_1 - \delta$.
\label{lemma:relaxation}
\end{restatable}

We now analyze the number of swaps required for \firstfairalgo\ to achieve an envy of $(2+\alpha)M$ for any $\alpha>0$, and hence provide a bound on its $CoF$.

\begin{restatable}{theorem}{thmnumiters}
Given that the minimum cost solution results in envy of $\mathcal{E}(S^*)$, the number of swaps required for \firstfairalgo\ to achieve an envy of $(2+\alpha)M$ for any $\alpha>0$, in the worst case, is 
$ \left\lfloor \frac{n}{2} \right\rfloor \left\lceil \log_2 \left( \frac{\mathcal{E}(S^*)-2M}{\alpha M}\right) \right\rceil$, which is $O\left(n\log\left(\frac{K}{\alpha}\right)\right)$.
\label{thm:num_iters}
\end{restatable}

\begin{proof}
Define one round as the set of swaps that reduces envy from $\epsilon$ to at most $\frac{\epsilon}{2}+M$. 
As \firstfairalgo\ performs a swap only when $\epsilon > 2M$, we have $\frac{\epsilon}{2}+M < \epsilon$. 
Further, 
denote envy at the beginning of a round $r$ to be $\epsilon^{(r-1)}$. So, at the beginning of the first round, if $C_1^0, C_2^0,..., C_n^0$ are the agents' costs in decreasing order, we have that the two extreme agents are $1$ and $n$, and $\epsilon^{(0)} = \mathcal{E}(S^*) = C_1 - C_n$.

Now, if at any round $r$, two agents $i,j$ are chosen by \firstfairalgo\ to perform a swap, then it must satisfy $C_j^{(r-1)} - C_i^{(r-1)} > 2M$, as the algorithm would have terminated otherwise. Thus, 
from
\Cref{lemma:relaxation}
and considering $2\delta = C_j^{(r-1)} - C_i^{(r-1)} - 2M$, we have 
$C_i^{(r)}  , C_j^{(r)} \in [C_i^{(r-1)} + \delta, C_j^{(r-1)} - \delta]$, i.e.,
\begin{align*}
& C_i^{(r)},C_j^{(r)} \in \left[\frac{C_i^{(r-1)} + C_i^{(r-1)}}{2} - M, \frac{C_i^{(r-1)} + C_j^{(r-1)}}{2} + M \right] .
\end{align*}

We mark agents $i, j$ as `adjusted in round $r$', and denote the set of all agents adjusted in round $r$ by $A_r$. Post swap, let us denote the two extreme agents as $s$ and $t$,
where $s$ is the agent with the current lowest cost, and $t$ is the agent with the current highest cost.
Then one of the following cases results: \\
[3pt]
\textbf{Case 1}
($s=i$, $t=j$): \firstfairalgo\ terminates without further swaps as $\left\vert C_t^{(r)}-C_s^{(r)} \right\vert \le 2M$ from 
\Cref{thm:envycbalance}.
\\
[3pt]
\textbf{Case 2(a)} ($s\in A_r$, $t\notin A_r$):
Let the last agent with whom $s$ swapped its path be $u$. Thus, at the time of the swap between agents $s$ and $u$, we must have cost of all agents in $\left[ C_{\min(s,u)}^{(r-1)}, C_{\max(s,u)}^{(r-1)} \right]$, where $C_{\min(s,u)}^{(r-1)} = \min \{ C_{s}^{(r-1)}, C_{u}^{(r-1)}\}$ and $C_{\max(s,u)}^{(r-1)}=\max \{ C_{s}^{(r-1)},C_{u}^{(r-1)}\}$. Hence, $C_{t}^{(r-1)} \leq C_{\max(s,u)}^{(r-1)}$. We get,
\allowdisplaybreaks
\begin{align*}
C_{t}^{(r-1)} - C_s^{(r)} 
& \leq 
C_{\max(s,u)}^{(r-1)} - C_s^{(r)} 
\\
& \leq C_{\max(s,u)}^{(r-1)} - \left(\frac{C_{\min(s,u)}^{(r-1)} + C_{\max(s,u)}^{(r-1)}}{2} - M\right) 
\\ &
\leq \frac{C_{\max(s,u)}^{(r-1)} - C_{\min(s,u)}^{(r-1)}}{2} + M
\leq \frac{\epsilon^{(r-1)}}{2} + M .
\end{align*} 
\textbf{Case 2(b)} 
($s\notin A_r$, $t\in A_r$):
 Let the agent with whom $t$ swapped its path be $v$. Similar to above, $C_{s}^{(r-1)} \geq C_{\min(t,v)}^{(r-1)}$ and hence,
 \begin{align*}
     C_{t}^{(r)}-C_s^{(r-1)} \leq \frac{\epsilon^{(r-1)}}{2} + M .
 \end{align*}
\textbf{Case 2(c)} ($s,t\in A_r$ and $s$ was adjusted before $t$): 
Here, $C_{\max(s,u)}^{(r-1)} \geq C_{\max(t,v)}^{(r-1)}$.
Also, we know that $C_{t}^{(r)} < C_{\max(t,v)}^{(r-1)}$. So,
\begin{align*}
C_{t}^{(r)} \!- C_s^{(r)}
& \!\leq C_{\max(t,v)}^{(r-1)} \!- C_s^{(r)}
\!\leq C_{\max(s,u)}^{(r-1)} \!- C_s^{(r)}
\!\leq \frac{\epsilon^{(r-1)}}{2} \!+\! M .
\end{align*}
The last inequality is obtained on similar lines as  Case 2(a).
\\
[3pt]
\textbf{Case 2(d)} ($s,t\in A_r$ and $t$ was adjusted before $s$):
Just as the expression in Case 2(c) reduces to that in Case 2(a), it can be shown that the expression here reduces to that in Case 2(b), implying 
\begin{align*}
    C_{t}^{(r)}-C_s^{(r)} \leq \frac{\epsilon^{(r-1)}}{2} + M .
\end{align*}

Since in all the sub-cases of Case 2, the envy reduces, from $\epsilon^{(r-1)}$ at the start of round $r$, to at most $\frac{\epsilon^{(r-1)}}{2}+M$, round $r$ concludes if we reach Case 2, and round $r+1$ begins.
\\
[3pt]
\textbf{Case 3} ($s,t\notin A_r$):
These two agents' paths are swapped, and $A_r = A_r \cup \{s,t\}$. 
\\
[3pt]
It should be noted that each round $r$ can enter Case 3 for at most $\left\lfloor \frac{n}{2} \right\rfloor$ swaps, otherwise, the algorithm is either terminating (Case 1) or entering the next round (Case 2).
Hence, the bound on envy becomes $\frac{\mathcal{E}(S^*)}{2}+M$ after at most $\left\lfloor \frac{n}{2} \right\rfloor$ swaps, $\frac{\frac{\mathcal{E}(S^*)}{2}+M}{2}+M = \frac{\mathcal{E}(S^*)}{4} + M(1+\frac{1}{2})$ after at most $2\left\lfloor \frac{n}{2} \right\rfloor$ swaps, and so on. 
In general, 
the bound on envy after at most $p\left\lfloor \frac{n}{2} \right\rfloor$ swaps
becomes 
\begin{align*}
 \frac{\mathcal{E}(S^*)}{2^p}+M \sum_{\gamma=1}^p \frac{1}{2^{\gamma-1}} 
 \;\text{, i.e., }\;
\frac{\mathcal{E}(S^*)}{2^p}+2M \left( 1-\frac{1}{2^p} \right) .
\end{align*}

We now obtain the value of $p$ for which the envy of $(2+\alpha)M$ is achieved, by equating  $\frac{\mathcal{E}(S^*)}{2^p} + 2M \left( 1-\frac{1}{2^p} \right)$ to $(2+\alpha)M$. 
We get $p = \left\lceil \log_2 \left( \frac{\mathcal{E}(S^*)-2M}{\alpha M} \right) \right\rceil$, as $p$ is an integer.
Hence, the number of swaps required in the worst case is 
$\left\lfloor \frac{n}{2} \right\rfloor \left\lceil \log_2 \left( \frac{\mathcal{E}(S^*)-2M}{\alpha M}\right) \right\rceil$.

Since the cost of any agent is bounded by $(K-1)M$, we have $\mathcal{E}(S^*) \leq (K-1)M$. 
Hence, the number of swaps in the worst case is bounded by $\left\lfloor \frac{n}{2} \right\rfloor \left\lceil \log_2 \left( \frac{(K-3)M}{\alpha M}\right) \right\rceil$, which is
$O\left( n \log \left( \frac{K}{\alpha} \right) \right)$.
\end{proof}

Recall that each swap in \firstfairalgo\ results in a maximum increase of $2M$ in the overall cost. 
Hence, from 
\Cref{thm:num_iters},
the increase in the overall cost due to the deviation from the minimum cost assignment is upper bounded by $2M\left\lfloor \frac{n}{2} \right\rfloor \left\lceil \log_2 \left( \frac{\mathcal{E}(S^*)-2M}{\alpha M}\right) \right\rceil$.
That is, in effect, \firstfairalgo\
provides 
bi-criteria additive approximation guarantees, namely, additive approximation of $(2+\alpha)M$ on envy and  $2M \left\lfloor \frac{n}{2} \right\rfloor \left\lceil \log_2  \left( \frac{\mathcal{E}(S^*)-2M}{\alpha M}\right) \right\rceil$ on cost.

Note that \firstfairalgo\ would perform any swaps only if $\mathcal{E}(S^*) > (2+\alpha)M$, which is possible only if the cost of at least one agent exceeds $(2+\alpha)M$. 
Further, it is possible to have $\mathcal{E}(S^*) > (2+\alpha)M$ only if $C(S^*) > (2+\alpha)M$.
If no swaps are performed, $CoF = 1$. 
The following result with regard to the bound on $CoF$ is immediate.

\begin{theorem}
The CoF of \firstfairalgo, for achieving an envy of $(2+\alpha)M$ for any $\alpha>0$, is bounded by $1 + \frac{2M\left\lfloor \frac{n}{2} \right\rfloor \left\lceil \log_2 \left( \frac{\mathcal{E}(S^*)-2M}{\alpha M}\right) \right\rceil}{C(S^*)}$, which is $O\left(n\log\left(\frac{K}{\alpha}\right)\right)$.
\label{thm:CoF}
\end{theorem}

That is, the $CoF$ bound of \firstfairalgo\ approximately increases with the logarithm of the ratio of the envy in the minimum cost assignment to the desired degree of envy.
It also increases with the maximum edge weight 
and the number of agents.
While the above expression of $CoF$ bound indicates a dependency on $n$, this dependency
is replaceable with that on the minimum edge weight, say $m$, if $m>0$. This is because, since $C(S^*) \geq m n
\geq 2 m \left\lfloor \frac{n}{2} \right\rfloor
$, the above $CoF$ bound is further bounded by $1 + \frac{M }{m } \left\lceil \log_2 \left( \frac{\mathcal{E}(S^*)-2M}{\alpha M} \right) \right\rceil$,
which is $O\left(\frac{M}{m}\log\left(\frac{K}{\alpha}\right)\right)$.
We now discuss the dependency of the CoF bound on instance-specific parameters.

\subsection{A Note on Dependency of the \texorpdfstring{$CoF$}{CoF} Bound on Instance-Specific Parameters}

Consider an instance where, in the minimum cost assignment, all the edges in the path of one of the agents have an arbitrarily small weight of $\gamma > 0$, while the paths of the other $(n-1)$ agents have edges with zero weight; 
an edge connecting a node on the non-zero-cost path with a node on a zero-cost path has a large weight of $M$, while an edge connecting a node on a zero-cost path with that on another zero-cost path has zero weight. For this instance, the minimum cost assignment leads to an overall cost as well as envy of $(K-1)\gamma$.

Consider $n = 2$ and $K \geq 2$.
If $K$ is even, the envy can be reduced to zero by making one swap at the middlemost stage as the cost of either of the agents would be $\left(\frac{K-2}{2}\right)\gamma+M$. 
If $K$ is odd, the envy can be reduced to $\gamma$ by making one swap at the middlemost stage, rounded to an integer; 
here the cost of one agent would be $\left\lfloor \frac{K-2}{2}\right\rfloor \gamma+M$ and of the other would be $\left\lceil \frac{K-2}{2}\right\rceil \gamma+M$. 
In both cases, the overall cost would be 
$(K-2)\gamma+2M$,
resulting in a CoF of
$
    \frac{(K-2)\gamma+2M}{(K-1)\gamma} ,
$
which could be arbitrarily high depending on $M$.
In general, an envy-minimizing assignment in such instances leads to a $CoF$ that 
is dependent on instance-specific paramaters.
A detailed explanation for $n \geq 3$ and $K \geq 2 n + 1$ is provided in 
\Cref{app:NoCoFbound}.
We, hence,
have the following observation.

\begin{observation}
It is infeasible to provide a constant instance-independent bound on $CoF$ for minimizing envy.
\label{obs:noinstind}
\end{observation}
\noindent Note that when $\alpha < \frac{\omega}{M}$, with $\omega>0$ being the minimum non-zero edge weight, then envy of \firstfairalgo\ would always be less then $2M$. This is because each swap leads to discrete changes in agents' cost due to discretized nature of edge weights.

\subsection{Extension to General FCMS Graphs}
\label{sec:unequal_stages}
In a general FCMS graph 
where some/all stages may have 
more nodes than the number of agents, consider the subgraph induced by the $n$ node-disjoint paths that result in the minimum total cost (obtained using the adaptation of Suurballe’s algorithm described earlier). Clearly, each of these $n$ paths contains exactly one node from each stage, resulting in the induced subgraph being a BFCMS graph.
Now, 
\firstfairalgo\ can be applied on this
induced BFCMS graph 
while ignoring all the nodes and edges not present in 
it.
So, the resultant envy and total cost of our algorithm in the given FCMS graph are same as that in the induced BFCMS graph. 
Also, by construction, a minimum cost assignment in the induced BFCMS graph is also a minimum cost assignment in the given FCMS graph. 
Hence, our analyses on envy bound and $CoF$ are directly applicable.

Note that envy-minimizing assignments in the given FCMS graph and the induced BFCMS graph could be different. However, our approach provides a solution with bounded envy while ensuring a bound on $CoF$ and is not aimed at minimizing envy, which is NP-hard
(Theorems~\ref{thm:nphard_numagents} and \ref{thm:nphard_numlevels}). 
Therefore, a solution obtained in the induced BFCMS graph is also a solution in the original FCMS graph.

\subsection{A Note on Minimax Share (MMS)}

 In  applications like supply chain, vehicle routing problem, etc., it is desirable to 
 minimize the maximum cost incurred by any agent. This is often termed \emph{makespan} in the scheduling literature. 
It can be shown that the difference between 
the maximum cost of an agent as provided by \firstfairalgo\ and that of the optimal assignment minimizing the maximum cost over the agents, is bounded by $O\left(M\log\left(\frac{K}{\alpha}\right)\right)$.
A detailed discussion is provided in 
\Cref{app:mms}.

\section{Experimental Evaluation}
\label{sec:experiments}

We now present an evaluation of the performance of the \firstfairalgo\ algorithm. 
For our experiments, we also employed a modified version of \firstfairalgo, called the Extended Dynamic Cost-Balance (\extfairalgo) algorithm, with the difference that we continue swapping while envy decreases. 
While the aforementioned adaptation of Suurballe’s
algorithm provides a minimum cost solution for a general FCMS graph, for the case of BFCMS graph, a minimum cost algorithm (\genmincostalgo) can be implemented using a relatively simpler approach which involves invoking Hungarian algorithm on each of the $K-1$ elementary 2-stage assignment problems; the details are provided in 
\Cref{app:seqhung}.
We validate our theoretical results through simulations, comparing them to results obtained from an Integer Linear Programming (ILP) formulation minimizing the cost of the solution subject to the condition that envy is no more than $2M$ for fair comparison, 
since our algorithms return a solution with envy less than or equal to $2M$. One could also run the ILP that minimizes the envy, but that would not be a fair comparison as such an ILP would result in significantly high cost solutions as compared to ours.
The ILP formulation is provided in 
\Cref{app:ILP}. 
We examined the impact of the number of agents and number of stages on the final envy ratio \big($\frac{\mathcal{E}(S)}{M}$\big), cost of fairness ($CoF$), and running time for the \firstfairalgo, \extfairalgo, \genmincostalgo, and ILP on a synthetic dataset to validate our theoretical findings.

All the experiments were run on Intel Xeon Gold $6246$R CPU @ $3.40$GHz with $384$GB RAM, and Python $3.6$.
To solve the ILP, Gurobi Optimizer v11 was used.
To generate the dataset, we created a collection of BFCMS graphs for $10$ agents, with the number of stages varying from $20$ to $80$ in increments of $5$. Additionally, we created a collection of BFCMS graphs with $40$ stages and the number of agents varying from $2$ to $20$ in increments of $2$. 
For each combination of the number of agents and stages, $500$ graphs were randomly generated. 
Note that the number of graph instances, on which our experiments were to be run, was limited due to the exceedingly high running time of ILP (which we later quantify in our observations).
The integer edge weights were generated uniformly at random within the range $[1, 30]$. The algorithms were executed on these BFCMS graphs, and their performance was recorded. Since \firstfairalgo\ does not perform any swaps if the envy in the minimum cost assignment is lower than $2M$, we 
employed rejection sampling
to restrict the set of graph instances to those where this envy exceeds $2M$, which is when \firstfairalgo\ performs any operations at all.

\subsection*{Observations}

\Cref{fig:envy} shows that
the envy ratio for \firstfairalgo\ is bounded by $2$, as guaranteed by our theoretical analysis. 
Consequently, the existence of a solution having envy of at most $2M$ is guaranteed, thus resulting in the envy ratio for ILP also being bounded by $2$.
The envy ratios for \firstfairalgo\ and ILP are close and do not see significant variation with the number of agents or stages, since their primary criterion is that envy is to be less than $2M$.
Since \extfairalgo\ continues 
swapping and 
reducing envy 
even post achieving an 
envy of under $2M$, its envy ratio is significantly lower.

\begin{figure}[t]
\begin{subfigure}[b]{0.24\textwidth}
    \centering\includegraphics[width=\linewidth]{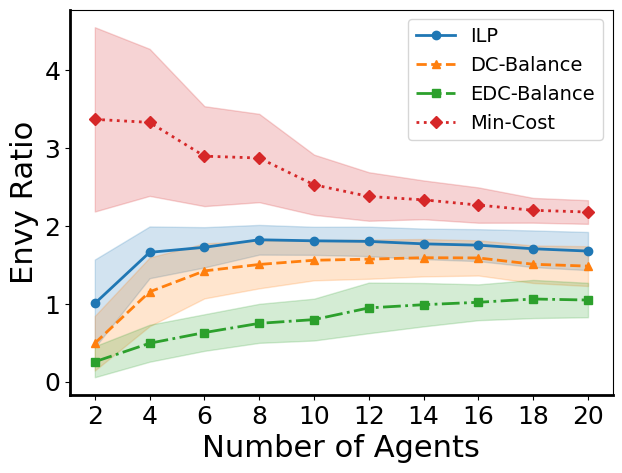}
    \caption{}
    \label{fig:envy_n}
  \end{subfigure}
  \begin{subfigure}[b]{0.23\textwidth}
    \centering \includegraphics[width=1.04\linewidth]{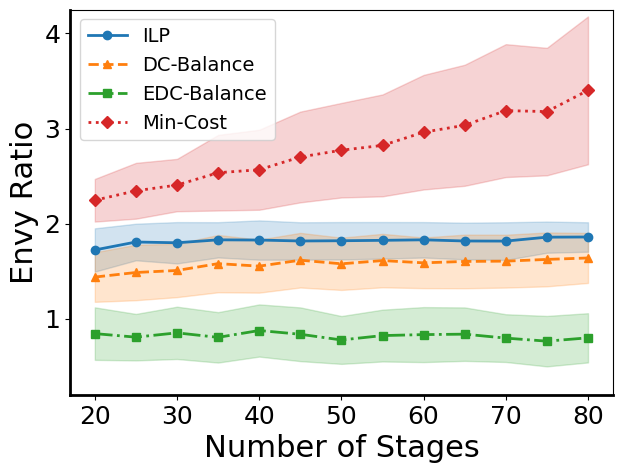}
    \caption{}
    \label{fig:envy_k}
  \end{subfigure}
  \captionsetup{skip=-12pt}
  \caption{Impact on the envy ratio 
  with respect to (a) the number of agents (given $K=40$) and (b) the number of stages (given $n=10$).
  }
  \label{fig:envy}
\end{figure}

\begin{figure}[t]
\begin{subfigure}[b]{0.23\textwidth}
    \centering\includegraphics[width=1.03\linewidth]{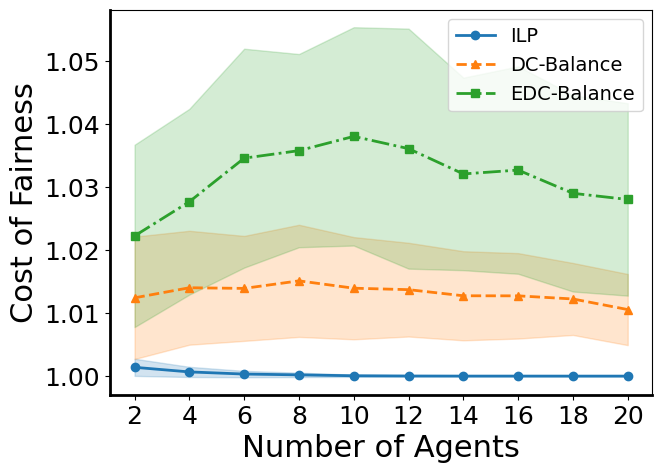}
    \caption{}
    \label{fig:cof_n}
  \end{subfigure}
  \begin{subfigure}[b]{0.23\textwidth}
    \centering \includegraphics[width=1.04\linewidth]{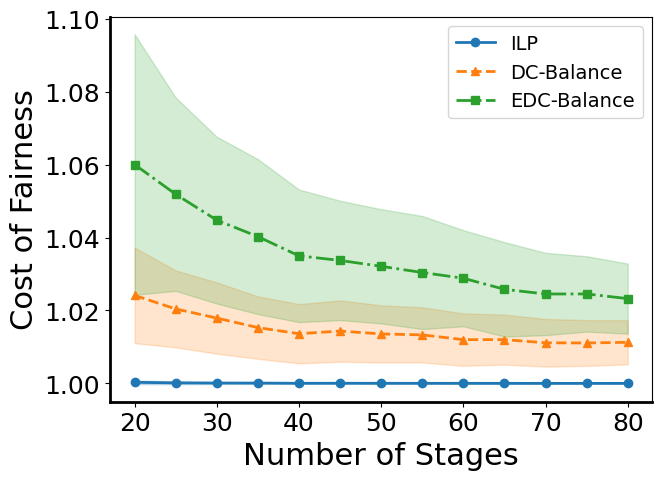}
    \caption{}
    \label{fig:cof_k}
  \end{subfigure}
  \captionsetup{skip=-12pt}
  \caption{Impact on $CoF$ with respect to (a) the number of agents (given $K=40$) and (b) the number of stages (given $n=10$).
  }
   \label{fig:cof}
\end{figure}

\Cref{fig:cof_n} shows that the dependency of CoF on the number of agents for \firstfairalgo\ is not as prominent. As explained after \Cref{thm:CoF}, the dependency of the CoF bound on $n$ is replaceable with that on the minimum edge weight $m$ (if $m>0$) and $M$.
As the edge weights for all the graph instances are 
drawn
from the set of integers in $[1,30]$, the values of $M$ and $m$ would be the same in almost all instances, thus resulting in the dependency of CoF on $n$ being not as prominent.
From \Cref{fig:cof_k}, the CoF in excess of $1$ 
(i.e., $\text{CoF}-1$) 
for \firstfairalgo\ seems to change inversely with the number of stages. This is inferable from \Cref{thm:CoF} to an extent, which shows that the bound on CoF in excess of $1$, 
decreases as an inverse function of $C(S^*)$ and hence also of $K$ (as $C(S^*)$ would be directly proportional to $K$, given that the edge weights are generated uniformly at random).
Considering also $\frac{\mathcal{E}(S^*)}{M}$ to increase approximately linearly with $K$ (as per the envy ratio plot for \genmincostalgo\ in \Cref{fig:envy_k}), the bound on CoF in excess of $1$ (\Cref{thm:CoF}) would be approximately proportional to $\frac{\log K}{K}$,
which closely resembles an inverse function of $K$.
The $CoF$ for the \extfairalgo\ algorithm is higher, as expected, while that for the ILP is the lowest owing to its formulation for minimizing cost.

 \Cref{fig:time} shows that the running time for ILP is several orders of magnitude higher than that of the other algorithms;  this follows from our discussion on the computational inefficiency of solving analogous ILPs in the literature. The running time for ILP increases rapidly with the number of stages, and even more so, exponentially with the number of agents. To put its running time in perspective, for $20$ agents and $40$ stages, \firstfairalgo\ takes $\sim$$0.02$ seconds on average, while ILP takes $\sim$$6000$ seconds which is slower by a factor of $\sim$$3 \times 10^5$.
Owing to the ILP approach being computationally exorbitant, a primary advantage of \firstfairalgo\
is its 
low computational cost while still providing bounds on 
envy and $CoF$.
From \Cref{fig:time_noILP} (which excludes ILP for better visualization of the other algorithms' running times), a trivial observation is that the running time of \extfairalgo\ exceeds that of \firstfairalgo, which in turn, exceeds that of \genmincostalgo.
A  notable observation is that the running time of \firstfairalgo\ increases linearly with the number of agents and logarithmically with the number of stages.
This is directly relatable to \Cref{thm:num_iters} which showed that the bound on the number of swaps required is a linear function of $n$ and logarithmic function of $K$.

\begin{figure}[t]
\begin{subfigure}[b]{0.24\textwidth}
    \centering\includegraphics[width=1.03\linewidth]{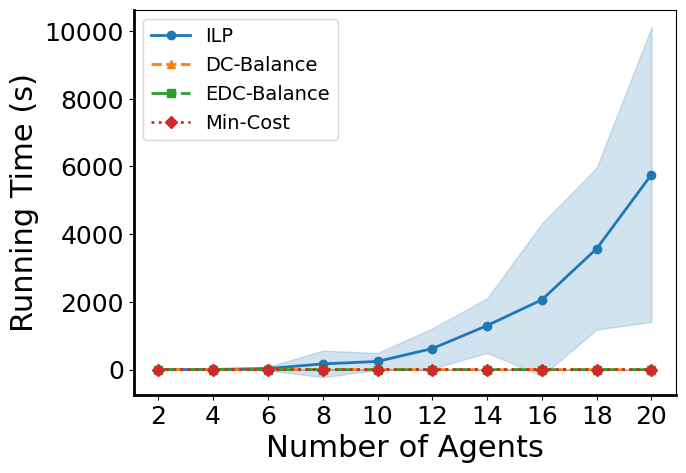}
    \caption{}
    \label{fig:time_n}
  \end{subfigure}
  \begin{subfigure}[b]{0.23\textwidth}
    \centering \includegraphics[width=1.07\linewidth]{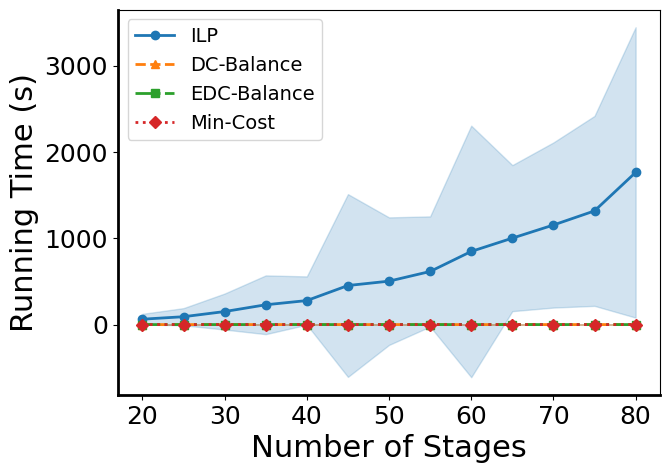}
    \caption{}
    \label{fig:time_k}
  \end{subfigure}
  \captionsetup{skip=-12pt}
  \caption{Impact on running time with respect to (a) the number of agents (given $K = 40$) and (b) the number of stages (given $n = 10$).
  }
   \label{fig:time}
\end{figure}

\begin{figure}[t]
\begin{subfigure}[b]{0.24\textwidth}
    \centering\includegraphics[width=1.03\linewidth]{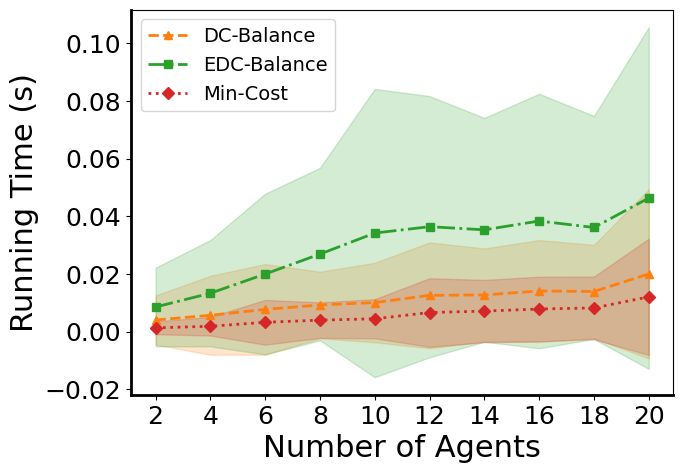}
    \caption{}
    \label{fig:time_noILP_n}
  \end{subfigure}
  \begin{subfigure}[b]{0.23\textwidth}
    \centering \includegraphics[width=1.07\linewidth]{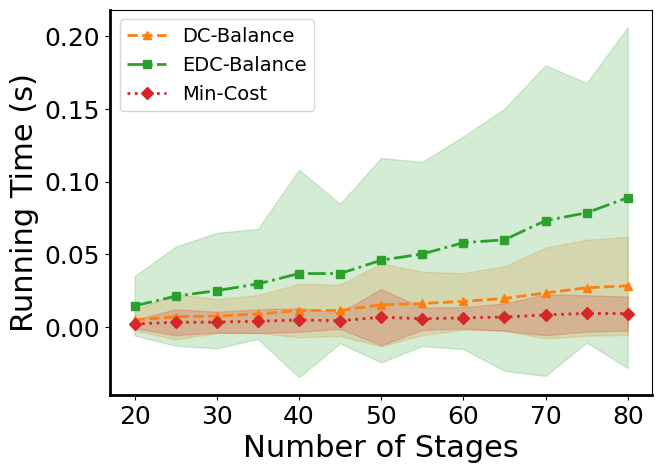}
    \caption{}
    \label{fig:time_noILP_k}
  \end{subfigure}
  \captionsetup{skip=-12pt}
  \caption{Impact on running time 
  for algorithms excluding ILP
  with respect to (a) the number of agents (given $K=40$) and (b) the number of stages (given $n=10$).
  }
   \label{fig:time_noILP}
\end{figure}

\section{Conclusion}
This paper introduced the problem of fair assignment on multi-stage graphs, where each agent needs to be assigned a distinct path such that the envy among them is mitigated. 
The paper showed that 
minimizing the envy among the agents is NP-hard,
and hence
offered poly-time algorithms that 
provide 
bi-criteria additive approximation guarantees on envy and cost.
Through experiments, we showed that the proposed solutions are far more efficient in terms of running time as opposed to the ILP approach
that is typically employed for analogous problems in the literature on fair scheduling.

As a future direction, one can extend the algorithms to constrained path assignment problems, where a node can only cater to a subset of agents, or after reaching a particular node in a stage, an agent can only go to a subset of nodes in the next stage. 
One can also look at the problem where the agents might have preference ordering over the nodes of the stages. For example, some workstations at a stage might be preferred by an agent over others. In this case, one will have to design a fairness metric considering agents' preferences
along with their costs. 
One can also look beyond the multi-stage graph and consider a fair disjoint paths problem on general graphs; this would pose significant challenges as general graphs lack a specific structure. For example, \cite{Lochet21} showed that even deciding whether an instance admits a disjoint shortest path solution is a W[1] hard problem parametrized by $n$ in a general graph. Therefore, finding a minimum or bounded envy solution on general graphs will be quite a challenging problem.

\begin{ack}

Vibulan J was supported by Uplink Internship Program of India Chapter of ACM SIGKDD. The project is also partially supported by  Anusandhan National Research Foundation, India, with grant number: MTR/2022/000818.

\end{ack}


\label{page:end}

\newpage
\appendix

\section*{APPENDICES}

\section{Proofs: NP-Hardness of Minimizing Envy}
\label{app:nphard}

\thmnphardnumagents*

\begin{proof}

Consider a reduction from Numerical 3-Dimensional Matching (N3DM).
In this problem, we are given disjoint sets $A = \{a_1,...,a_n\}$, $B = \{b_1,...,b_n\}$, $C = \{c_1,...,c_n\}$ of integer elements 
such that $\sum A + \sum B + \sum C = nt$,
where $\sum A$ denotes $\sum_{a \in A} a$. The objective is to determine if it is possible to partition $A \cup B \cup C$ into $n$ disjoint sets $S = S_1,...,S_n$ such that each $S_i$ contains exactly one element from each of $A$, $B$ and $C$ and 
$\sum_{e\in S_i} e = t$. 
This problem is known to be NP-complete in the strong sense \cite{garey1979computers}.

For each instance of N3DM, let us construct a BFCMS graph $\mathcal{G} = (V, E)$ with the vertex set $V$ partitioned into $4$ disjoint subsets (or stages) $\{V_1, V_{2}, V_{3}, V_{4}\}$, each containing $n$ nodes. Every node from the stage $V_i$ is linked to every node in the stage $V_{i+1}$. Each node in stages $2, 3, 4$ are assigned a label corresponding to an element in $A, B, C$ respectively. Assign edge weights such that every edge ending at node with label of integer element in $A \cup B \cup C$ has weight equal to that integer element.
We show that the envy of the envy minimizing assignment is $0$, if and only if, there exists a desired partition in the original N3DM instance.

$(\Longrightarrow)$
If there is a partition $P = \{a^\prime_1, b^\prime_1, c^\prime_1\},...,\{a^\prime_n, b^\prime_n, c^\prime_n\}$ satisfying the constraints in the original N3DM instance, 
construct $n$ node-disjoint paths starting from any node in $V_1$ not in any paths constructed so far and passing through nodes labeled $a^\prime_i, b^\prime_i, c^\prime_i$. Since these paths are node-disjoint, this is a valid assignment on the BFCMS graph and each path has cost $t$. Therefore, the envy of this assignment is $0$.

$(\Longleftarrow)$
If the envy of the envy minimizing assignment of the constructed BFCMS graph is $0$, then the minimum envy assignment consists of $n$ node-disjoint paths, say through nodes $(v_1, a^\prime_1, b^\prime_1, c^\prime_1),\ldots,$ $(v_n, a^\prime_n, b^\prime_n, c^\prime_n)$, each of cost $t$. This is because the total cost of any assignment is $nt$. Let $S^\prime_i = \{a^\prime_i, b^\prime_i, c^\prime_i\}$. Observe that $S^\prime = S^\prime_1,...,S^\prime_n$ is a desired partition in the N3DM instance as elements in each $S^\prime_i$ sum up to $t$.

Therefore, finding the envy of the envy minimizing assignment is NP-hard for a BFCMS Graph $\mathcal{G}$ with $K = 4$ since otherwise, N3DM could be solved by constructing the corresponding BFCMS graph and checking if the envy of the envy minimizing solution is $0$.
Consequently, finding the envy minimizing assignment is also NP-hard for a BFCMS Graph $\mathcal{G}$ with $K = 4$.

For an arbitrary BFCMS Graph $\mathcal{G}$ with $K > 4$, finding an envy minimizing assignment is also NP-hard. Otherwise, we can find an envy minimizing solution with 4 stages by adding more stages where every edge ending at nodes at the stage has cost $0$.
\end{proof}

\thmnphardnumlevels*

\begin{proof}

Consider a reduction from Multi-Way Number Partitioning, sometimes referred to as `Multi-Processor Scheduling'. The problem is known to be NP-complete \cite{garey1979computers}. The problem is parametrized by a positive integer $n \ge 2$.  The input to the problem is a multiset $S$ of numbers, whose sum is $nT$. The problem is to decide whether $S$ can be partitioned into $n$ subsets such that the sum of each subset is exactly $T$.

For each instance of $n$-way number partition, construct a BFCMS graph $\mathcal{G}$ with $n$ agents and $\lvert S \rvert + 1$ stages. At each stage greater than $1$, choose any one node and assign a label corresponding to any element in $S$ such that every element in $S$ is assigned as a label to exactly one node. Assign edge weights such that every edge ending at a labeled node has a weight equal to the value of the label of the node. For every other edge, assign an edge weight of $0$.
We show that the envy of the envy minimizing assignment is $0$, if and only if, there exists a desired partition in the original $n$-way number partition instance.

$(\Longrightarrow)$
If there is a partition of $S$ into $n$ subsets such that the sum of each subset is exactly $T$, then we can construct an assignment in a BFCMS graph as follows.  For each such subset, construct paths starting from any node in the first stage not in any paths constructed so far. For each subsequent stage, if the label of the labeled node in that stage is in the subset, let the path pass through the node; otherwise, let the path pass through any other node in the stage not in any paths constructed so far. Observe that each constructed path corresponds to some subset in the sense that it passes through labeled nodes that have the label of the elements in the subset, and has path cost equal to the sum of the elements in the subset. Since all the $n$ node-disjoint paths have cost $T$, the envy of this assignment is $0$.

$(\Longleftarrow)$
If the envy of the envy minimizing assignment is $0$, then the assignment consists of $n$ node-disjoint paths. Each path has cost $T$. For each path, consider the subset of elements of $S$ that are labels of all the labeled nodes the path passes through. These subsets are mutually exclusive, and the union of the subsets is $S$. Moreover, the sum of each subset is exactly $T$. Therefore there exists a desired partition in the original $n$-way number partition instance.

Therefore, finding the envy of envy minimizing assignment is NP-hard for a BFCMS graph with a variable number of stages and exactly $n$ agents for $n \ge 2$. As otherwise $n$-way number partition can be solved by constructing the corresponding BFCMS graph and checking if the envy of the envy minimizing solution is $0$. 
Consequently, finding the envy minimizing assignment is also NP-hard for a BFCMS Graph $\mathcal{G}$ with $n$ agents.
\end{proof}

\section{Proof of Bounds on the New Costs After a Swap}

\label{app:relaxation}

\lemmarelaxation*

\begin{proof}

From the inequalities in (\ref{eq:ineq_n}) of the main paper, we know that the choice of stage $i$ (starting from which the nodes allocated to the considered pair of agents are swapped) ensures that
\begin{align*}
& \frac{\epsilon}{2} - M \le C_1'' - C_n'' \le \frac{\epsilon}{2} 
\;\;\;
\text{ and } \;\;\;
\frac{\epsilon}{2} < C_1' - C_n' \le \frac{\epsilon}{2}+M .
\end{align*}
We use these,
alongwith
$-M \leq w_1-w_2 \leq M, \,  \forall w_1, w_2 \in \{x, x', y, y'\} \; (\because x, x', y, y' \in [0,M])$,
in the following derivation.
\allowdisplaybreaks
\begin{align*}
\text{(a) }\;
C_1^{new} 
&=  C_1'' + x' + C_n - C_n'
\\
&\le \left( \frac{\epsilon}{2} + C_n'' \right) + x'+ C_n - (C_n'' + y)
\\&= C_n + \frac{\epsilon}{2} + (x'-y)
\\
&\le C_n + \frac{\epsilon}{2} + M 
\\
&= C_1 - \frac{\epsilon}{2} + M 
\\
&= C_1 - \delta .
\\[0.5em]
\text{(b) }\;
C_1^{new} 
&=  C_1'' + x' + C_n - C_n'
\\&= (C_1'  - x) + x'+ C_n  - C_n'
\\
&= C_n + (C_1' - C_n') + (x' - x)
\\&\ge C_n + \frac{\epsilon}{2} - M 
\\
&= C_n + \delta .
\\[0.5em]
\text{(c) }\;
C_n^{new} 
&=  C_n'' + y' + C_1 - C_1'
\\&= C_n'' + y' + C_1 - (C_1'' + x)
\\
&= C_1 + (C_n'' - C_1'') + (y'  - x)
\\&\le C_1 - \frac{\epsilon}{2} + M 
\\
&= C_1 - \delta .
\\[0.5em]
\text{(d) }\;
C_n^{new} 
&=  C_n'' + y' + C_1 - C_1'
\\&= C_n'' + y'+ (C_n + \epsilon) - (C_1'' + x)
\\
&\ge \left(C_1'' - \frac{\epsilon}{2} \right) + y' + C_n + \epsilon - C_1''  - x
\\&= C_n + \frac{\epsilon}{2} + (y'-x)
\\
&\ge C_n + \frac{\epsilon}{2} - M 
\\
&= C_n + \delta .
\end{align*}

The above, combined, give the result.
\end{proof}

\section{Dependency of the \texorpdfstring{$CoF$}{CoF} Bound on Instance- Specific Parameters for \texorpdfstring{$n \geq 3$}{n ≥ 3} and \texorpdfstring{$K \geq 2 n + 1$}{K ≥ 2n + 1}}
\label{app:NoCoFbound}

Refer to 
\Cref{fig:CoF_Mn}.
For any $n \geq 3$ and $K \geq 2 n + 1$, it can be seen, by considering all possible assignments via brute force, that in the envy-minimizing assignment, $(K-1) \mod \left\lfloor \frac{K-1}{n} \right\rfloor$ agents incur a cost of $2M + \left\lceil \frac{K-1}{n}-1 \right\rceil \gamma$ and the remaining agents incur a cost of $2M + \left\lfloor \frac{K-1}{n}-1 \right\rfloor \gamma$, thus resulting in an envy of $0$ if $K-1$ is a multiple of $n$ and an envy of $\gamma$ otherwise.

For ease of exposition, let us consider the case where $K-1$ is a multiple of $n$, wherein each of the $n$ agents in the envy-minimizing assignment incurs a cost of $2M + \left( \frac{K-1}{n}-1 \right) \gamma$. Hence, the envy is $0$, and the overall cost is $2Mn + \left( K-1-n \right) \gamma$.
This expression comes from the fact that in order to minimize the envy, each of the agents,
which had a zero-cost path in the minimum cost assignment,
is forced to visit a part of the non-zero cost path (i.e., a segment or at least a node); for doing so, each agent has to incur an additional cost of $M$ for visiting a part of the non-zero cost path and a further additional cost of $M$ for returning to its original path. On the other hand, the one agent, which had a non-zero-cost path in the minimum cost assignment, now travels the same length (as each of the other agents) of the path, which had edges with weight $\gamma$, and otherwise travels zero-weight edges which originally belonged to the paths of agents that incurred zero cost in the minimum cost assignment; for doing so, the agent has to incur an additional cost of $M$ for diverting from its original path and a further additional cost of $M$ for returning.
It should be further noted that each agent visits a part of the non-zero cost path at least once in order to minimize the envy (simultaneously, some other agent diverts from the non-zero cost path), and moreover, exactly once in order to minimize the envy while also keeping the overall cost to a minimum. 

Since all the $n$ agents incur a total cost of $2M$ for diverting from and returning to their original paths, the contribution to the overall cost, of the edges traversed in the minimum envy assignment, is $2Mn$.
Also, note that whenever an agent diverts from the non-zero cost path after a stage $i$, the edge having weight $\gamma$ between stages $i$ and $i+1$ is not traversed by any agent; this happens $n$ times as each agent diverts once.
So, the contribution to the overall cost, of the parts of the non-zero cost path traversed in the minimum envy assignment, is $(K-1)\gamma - n\gamma$.
Hence, for the case where $K-1$ is a multiple of $n$, the resulting $CoF$ is 
\begin{align*}
\frac{2Mn + ( K-1-n ) \gamma}{(K-1)\gamma} \approx 1+\frac{2Mn}{(K-1)\gamma} 
\;\; \text{($\because \gamma << 2M$)} .
\end{align*}

\begin{figure}[t]
    \centering
\includegraphics[width=0.65\linewidth]{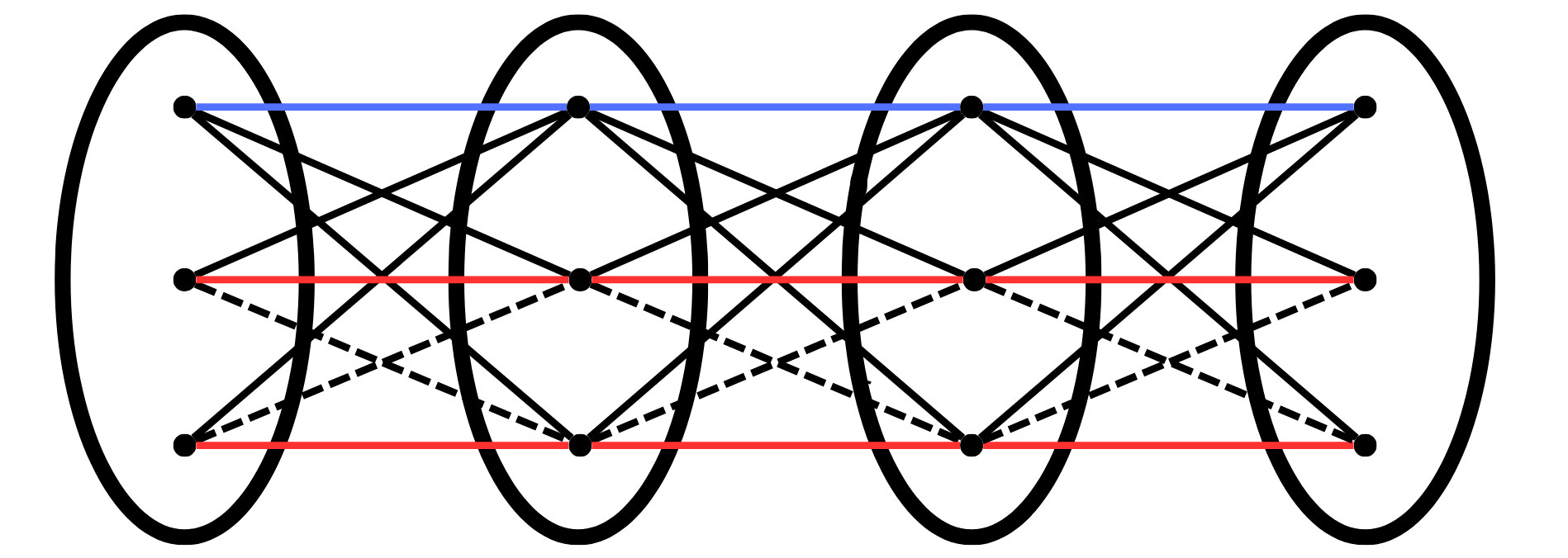}
    \caption{In the given example instance, edges marked blue have a weight of $\gamma$ and constitute the non-zero cost path of one of the agents in the minimum cost assignment. Edges marked red have a weight of zero and constitute the zero cost paths of the other $(n-1)$ agents in the minimum cost assignment.
    Edges denoted by solid black lines have a weight of $M$ and connect a node on the non-zero-cost path with a node on a zero-cost path.
    Edges denoted by dashed black lines have a weight of zero and connect a node on a zero-cost path with that on another zero-cost path.}
    \label{fig:CoF_Mn}
\end{figure}

For the case where $K-1$ is not a multiple of $n$, from the explanation provided earlier that $(K-1) \mod \left\lfloor \frac{K-1}{n} \right\rfloor$ agents incur a cost of $2M + \left\lceil \frac{K-1}{n}-1 \right\rceil \gamma$ and the remaining agents incur a cost of $2M + \left\lfloor \frac{K-1}{n}-1 \right\rfloor \gamma$, the $CoF$ can be obtained by dividing the overall cost in the minimum envy assignment by $(K-1)\gamma$, i.e., 

\begin{small}
\begin{align*}
\frac{1}{(K-1)\gamma}
\Biggl( & \left( (K\!-\!1) \!\!\!\!\! \mod \!\!\! \left\lfloor \frac{K-1}{n} \right\rfloor \right) \! \left( 2M \!+\! \left\lceil \frac{K-1}{n}\!-\!1 \right\rceil \gamma \right) \\ \!+\! & \left( n \!-\! (K\!-\!1) \!\!\!\!\! \mod \!\!\! \left\lfloor \frac{K-1}{n} \right\rfloor \right) \! \left( 2M \!+\! \left\lfloor \frac{K-1}{n}\!-\!1 \right\rfloor \gamma \right) \Biggr)
 .
\end{align*}
\end{small}

Thus, in general, an assignment that minimizes envy in such instances leads to a $CoF$ that increases with $M$ and $n$. Hence, 
it is infeasible to provide a constant instance-independent bound on $CoF$.

\section{A Note on Minimax Share (MMS)}

\label{app:mms}

In applications like supply chain, vehicle routing problem, etc., it is desirable to 
minimize the maximum cost 
of
any agent. This is often termed \emph{makespan} in the scheduling literature.

\begin{definition}[Minimax Share]
An agent’s minimax share cost is given by: $\text{MMS}(\mathcal{G}) = \min_{S \in \mathcal{F}} \max_{j \in [n]} C(P_j),$ where \(S = \{P_1,...,P_n\}\) is a solution in the set of all valid solutions \(\mathcal{F}\). Given a BFCMS graph \(\mathcal{G}\), assignment \(S\) is MMS-Fair if \(C(P_i) \leq MMS(\mathcal{G})\) for every \(i \in [n]\).
\end{definition}

For an arbitrary fully connected multi-stage graph with a given $K \ge 4$ and arbitrary $n$, or that with a given $n \ge 2$ and arbitrary $K$, finding an MMS-Fair assignment is NP-hard. The proof is similar to \Cref{thm:nphard_numlevels}. We now provide a bound on maximum cost of any agent in the solution returned by \firstfairalgo.

\textbf{Result.}
If the agents' costs in the solution returned by \firstfairalgo\ are arranged in descending order, say $C_1,..., C_n$,
we have $C_1-C_i \leq (2+\alpha)M, \forall i \in \{2,...,n\}$.
The difference between the maximum cost over agents ($C_1$) and the mean is
\begin{align*}
\allowdisplaybreaks
&C_1 - \frac{\sum_{i=1}^n C_i}{n} = \frac{\sum_{i=2}^n (C_1-C_i)}{n} \leq \frac{(n-1)(2+\alpha)M}{n}\\
&\implies
C_1 \leq \frac{\sum_{i=1}^n C_i + (n-1)(2+\alpha)M}{n} .
\end{align*}

\noindent
Therefore, the cost of an agent with the maximum cost, in a fair assignment returned by \firstfairalgo\ algorithm, is bounded by
\[
   \frac{C(S^*) + 2M\left\lfloor \frac{n}{2} \right\rfloor \left\lceil \log_2 \left( \frac{\mathcal{E}(S^*)-2M}{\alpha M}\right) \right\rceil + (n - 1)(2+\alpha)M}{n} .
\]

Note that the optimal MMS cost can be no less than $\frac{C(S^*)}{n}$. Therefore, the difference between the maximum cost of an agent as provided by \firstfairalgo\ and that of the optimal assignment minimizing
the maximum cost over the agents, is bounded by $O\left(M\log\left(\frac{K}{\alpha}\right)\right)$
(since $\alpha$ is a small quantity in general, implying $2+\alpha \in O(1)$).

\section{Sequential Hungarian Algorithm}

\label{app:seqhung}

We introduce an algorithm to find the optimal solution, which forms the base for subsequent fair solutions.
Consider a 2-stage graph problem 
with a total of $n^2$ edges between the first two stages of the depicted graph. 
The objective is to select $n$ edges from the available $n^2$ edges, ensuring no two edges share a common node and minimizing the overall sum of the selected edges' costs. 
This assignment problem for the 2-stage graph can be efficiently addressed using the Hungarian algorithm \cite{kuhn1955}.

A general assignment problem on BFCMS  with $K$ stages comprises a consecutive sequence of $K-1$ elementary 2-stage sub-problems.
To find the minimum cost solution, the Sequential Hungarian (\mincostalgo) algorithm 
[\Cref{alg:seqhung}]
iteratively invokes the Hungarian algorithm on each of these $K-1$ elementary 2-stage problems sequentially.

\begin{algorithm}[ht]
\DontPrintSemicolon

  \KwInput{A multi-stage graph $\mathcal{G} (\{V_j\}_{j=1}^K, E)$} 
  \KwOutput{\mbox{Solution $S=\{P_1,\ldots,P_n\}$ with $P_i = \{p_i^j\}_{j=1}^{K-1}$}}
  
     $S \gets \emptyset$

    \For{each $j$ in $[K-1]$}
    {
         $\{p_i^j\}_{i=1}^n \gets$ Hungarian(stage $j$, stage $j+1)$
         
         $S.\text{append}(\{p_i^j\}_{i=1}^n)$
    }
\caption{\mincostalgo}
\label{alg:seqhung}
\end{algorithm}

It is easy to see that the \mincostalgo\ algorithm 
returns a solution $S$ that satisfies $C(S) = C(S^*)$.
To show this, observe that
the total cost of solution $S=\{P_1,\ldots,P_n\}$ returned by \mincostalgo, with $p_i^j$ representing the edge of $P_i$ from stage $j$ to $j+1$,  is 
\begin{align*}
    C(S) = \sum_{i=1}^n\sum_{e\in P_i}w_e = \sum_{i=1}^n\sum_{j=1}^{K} w_{p_i^j} = \sum_{j=1}^{K}\sum_{i=1}^n w_{p_i^j} .
\end{align*}

For each $j$, \mincostalgo\ chooses the assignment that minimizes the sum of weights for all agents. Therefore, $\sum_{j=1}^{K}\sum_{i=1}^n w_{p_i^j} \le \sum_{j=1}^K\sum_{i=1}^n w_{e_i^j} ,$ 
where $e_i^j$ could be any arbitrary edge from stage $j$ to $j+1$. However, $\sum_{j=1}^K\sum_{i=1}^n w_{e_i^j}$ represents any possible solution $S \in \mathcal{F}$, thus we have $C(S) \le C(S')\ \forall S' \in \mathcal{F}$, i.e., $C(S) = C(S^*)$.

\newpage

\section{ILP Formulation}
\label{app:ILP}

The following ILP formulation minimizes cost in the BFCMS problem with the constraint that envy $\leq 2M$.

\begin{table}[ht]
\label{tab:elementsILP}
\centering
\caption{Elements of ILP}
\begin{tabular}{lp{0.35\textwidth}}
\toprule
\textbf{Sets} & \textbf{Description} \\
\midrule
$\mathcal{V}$ & Set of agents \\
$\mathcal{N}$ & Set of nodes \\
\midrule
\textbf{Variable} & \textbf{Description} \\
\midrule
$x_{ijv}$ & $x_{ijv} = 1$ if agent $v \in \mathcal{V}$ goes from \par node $i \in \mathcal{N}$ to node $j \in \mathcal{N}$ and 0 otherwise \\
\midrule
\textbf{Parameters} & \textbf{Description} \\
\midrule
$c_{ij}$ & Cost of going from node $i \in \mathcal{N}$ to node $j \in \mathcal{N}$ \\
$k(i)$ & Stage to which node $i$ belongs \\
M & Maximum edge weight \\
\bottomrule
\end{tabular}
\end{table}

\begin{equation*}
    \text{Minimize } \sum_{v \in \mathcal{V}} \sum_{i \in \mathcal{N}} \sum_{j \in \mathcal{N}} c_{ij} x_{ijv} 
\end{equation*}

\textbf{subject to}

\begin{flalign*}
    0 \le x_{ijv} \le 1, && \forall i \in \mathcal{N}, \forall j \in \mathcal{N}, \forall v \in \mathcal{V} 
    \\ 
    x_{ijv} = 0, && \forall j \in \mathcal{N}, \forall i \in \mathcal{N}, k(j) \neq k(i) + 1 
    \\
    \sum_{j \in \mathcal{N}} \sum_{v \in \mathcal{V}} x_{ijv} = 1, && \forall i \in \mathcal{N}, 1 \le k(i) \le K - 1 
    \\ 
    \sum_{i \in \mathcal{N}} \sum_{v \in \mathcal{V}} x_{ijv} = 1, && \forall j \in \mathcal{N}, 2 \le k(j) \le K 
    \\ 
    \sum_{\substack{i \in \mathcal{N} \\ k(i) = 1}} \sum_{j \in \mathcal{N}} x_{ijv} = 1, && \forall j \in \mathcal{N}, \forall v \in \mathcal{V} 
    \\ 
    \sum_{i_1 \in \mathcal{N}} x_{i_1 j v} = \sum_{i_2 \in \mathcal{N}} x_{j i_2 v}, && \forall v \in \mathcal{V}, \forall j \in \mathcal{N}, 2 \le k(j) \le K - 1 
\end{flalign*}
\vspace{-5mm}
\begin{flalign*}
    \sum_{i, j \in \mathcal{N}} c_{ij} x_{ijv_1} - \sum_{i, j \in \mathcal{N}} c_{ij} x_{ijv_2} \le 2M, &&\forall v_1, v_2 \in \mathcal{V}, v_1 \neq v_2 
\end{flalign*}

\end{document}